\documentclass[manuscript,nonacm,11pt]{acmart}

\newcommand{\bblocal}{\ensuremath{\mathsf{BB}}}
\newcommand{\bb}{\bblocal}

\DeclareMathOperator*{\ICost}{\mathsf{ICost}}
\newcommand{\prtn}{Z}

\usepackage{fullpage}

\usepackage{amsmath,amsthm,amsbsy}
\usepackage{microtype}
\usepackage{algorithm}
\usepackage{setspace}
\usepackage{color}
\usepackage{xcolor}
\usepackage[noend]{algpseudocode}
\usepackage{multirow}
\usepackage{paralist}
\usepackage{xspace}

\renewcommand{\Pr}{\operatorname*{\textbf{\textup{Pr}}}}
\DeclareMathOperator*{\II}{\textbf{\textup{I}}}
\DeclareMathOperator*{\HH}{\textbf{\textup{H}}}
\DeclareMathOperator*{\EE}{\textbf{\textup{E}}}

\DeclareMathOperator*{\DD}{\textbf{\textup{D}}}
\renewcommand{\geq}{\geqslant}
\renewcommand{\ge}{\geqslant}
\renewcommand{\leq}{\leqslant}
\renewcommand{\le}{\leqslant}
\newcommand{\set}[1]{\{ #1 \}}

\newcommand{\lb}{\left}
\newcommand{\rb}{\right}

\DeclareMathOperator{\poly}{\ensuremath{\mathrm{poly}}}

\makeatletter
\newtheorem*{rep@theorem}{\rep@title}
\newcommand{\newreptheorem}[2]{%
\newenvironment{rep#1}[1]{%
\def\rep@title{#2 \ref{##1}}%
\begin{rep@theorem}[restated]}%
{\end{rep@theorem}}}
\makeatother
\newreptheorem{lemma}{Lemma}
\newreptheorem{theorem}{Theorem}
\newreptheorem{corollary}{Corollary}

\newboolean{short}
\setboolean{short}{false}
\newcommand{\onlyShort}[1]{\ifthenelse{\boolean{short}}{#1}{}}
\newcommand{\onlyLong}[1]{\ifthenelse{\boolean{short}}{}{#1}}

\theoremstyle{definition}

\newtheorem{definition}{Definition}
\theoremstyle{plain}
\newtheorem{lemma}{Lemma}
\newtheorem{theorem}{Theorem}

\newtheorem{corollary}{Corollary}
\theoremstyle{plain}

\title[]{Robust Lower Bounds for Graph Problems in the Blackboard Model of Communication}
\author{Christian Konrad}
\affiliation{\institution{Department of Computer Science, University of Bristol} \country{}}
\email{christian.konrad@bristol.ac.uk}
\author{Peter Robinson}
\affiliation{\institution{Department of Computer Science, City University of Hong Kong} \country{}}
\email{peter.robinson@cityu.edu.hk}
\author{Viktor Zamaraev}
\affiliation{\institution{Department of Computer Science, University of Liverpool} \country{}}
\email{viktor.zamaraev@liverpool.ac.uk}

\copyrightyear{2020}
\acmYear{2020}
\setcopyright{acmlicensed}\acmConference{}{}{}
\acmBooktitle{}
\acmPrice{}
\acmDOI{}
\acmISBN{}

\begin{document}
\begin{abstract}
We give lower bounds on the communication complexity of graph problems in the multi-party blackboard model. In this model, the edges of an $n$-vertex input graph are partitioned among $k$ parties, who communicate solely by writing messages on a shared blackboard that is visible to every party. We show that any {\em non-trivial} graph problem on $n$-vertex graphs has blackboard communication complexity $\Omega(n)$ bits, even if the edges of the input graph are randomly assigned to the $k$ parties. We say that a graph problem is non-trivial if the output cannot be computed in a model where every party holds at most one edge and no communication is allowed. Our lower bound thus holds for essentially all key graph problems relevant to distributed computing, including \textsf{Maximal Independent Set} (MIS), \textsf{Maximal Matching}, \textsf{($\Delta+1$)-coloring}, and \textsf{Dominating Set}. In many cases, e.g., \textsf{MIS}, \textsf{Maximal Matching}, and \textsf{$(\Delta+1)$-coloring}, our lower bounds are optimal, up to poly-logarithmic factors.

\end{abstract}

\maketitle

\section{Introduction}
The multi-party blackboard model of communication is a well-established model of distributed computation \cite{dko12,pvz12,bo15,vww20}. In this model, the input is split among $k$ parties, who collectively need to solve a joint problem by communicating with each other. However, as opposed to point-to-point message passing models where pairs of parties have private communication links, the parties communicate solely via a shared blackboard that is visible to all parties. The objective is to design communication protocols with minimal {\em communication complexity}, i.e., protocols that instruct the parties to write as few bits onto the shared blackboard as possible. The model is asynchronous, and the order in which the parties write onto the blackboard is solely determined by the communication protocol.

In this paper, we initiate the study of  fundamental graph problems in the blackboard model. Given an input graph $G=(V, E)$ with $n = |V|$, we consider an edge-partition model, where the edges of $G$ are partitioned among the $k$ parties either adversarially or uniformly at random.
Our main result is a lower bound that applies to a large class of graph problems and that is tight (up to poly-logarithmic factors) in many cases:

\vspace{0.15cm}

\begin{center}
\fbox{\begin{minipage}{0.95\textwidth}
\begin{theorem}[simplified]\label{thm:main} Every {\em non-trivial} graph problem has (randomized) blackboard communication complexity of $\Omega(n)$ bits, even if the edges of the input graph are partitioned randomly among the parties.
\end{theorem}\end{minipage}
}
\end{center}

\vspace{0.15cm}

Informally, a problem is non-trivial if it does not admit a deterministic protocol where every party holds at most one edge of the input graph and no communication occurs at all (see Section~\ref{sec:lb} for a precise definition). Intuitively, the knowledge of at most one edge is not enough for the parties to solve any interesting graph problem, and, as proved in this paper, most graph problems, including \textsf{Minimum Vertex Cover}, \textsf{Minimum Dominating Set}, \textsf{Maximal/Maximum Matching}, \textsf{Maximal Independent Set}, and \textsf{$(\Delta+1)$-coloring}, are therefore non-trivial.

In many cases, our lower bound is tight up to poly-logarithmic factors, even if the edge partitioning is adversarial. For \textsf{Maximal Matching}, obtaining a blackboard communication protocol with communication cost $O(n \log n)$ is trivial. For others, such as \textsf{Maximal Independent Set} and \textsf{$(\Delta+1)$-coloring}, protocols with communication cost $O(n \poly \log n)$ can be derived from known data streaming algorithms. We refer the reader to Table~\ref{tab:results} for an overview of upper bounds that match our lower bounds up to $\poly \log n$ factors. These upper bounds will be discussed in Section~\ref{sec:upper-bounds}.

\begin{table}[h]
    \centering
    \begin{tabular}{|l|l|l|}
    \hline
    \textbf{Problem} & \textbf{CC Upper Bound} & \textbf{Reference} \\
    \hline
        Maximal Matching & $O(n \log n)$ & trivial  \\
        $(1-\epsilon)$-approx. Maximum Matching (general graphs) & $O(n \log n (\frac{1}{\epsilon})^{O(\frac{1}{\epsilon})})$ & McGregor \cite{m05} \\
        $(1-\epsilon)$-approx. Maximum Matching (bipartite graphs) & $O(n (\log n + \frac{1}{\epsilon} \log(\frac{1}{\epsilon})) \frac{1}{\epsilon^2})$ & Assadi et al. \cite{alt21} \\
        Maximal Independent Set & $O(n \poly \log n)$ & Ghaffari et al. \cite{ggkmr18} \\
        ($\Delta+1$)-coloring & $O(n \log^4 n)$ & Assadi et al. \cite{ack19} \\
        \hline
    \end{tabular}
    \caption{Matching (up to poly-logarithmic factors) upper bounds to our $\Omega(n)$ lower bound. }
    \label{tab:results}
\end{table}

\subsection{Related Work}
\subsubsection*{Multi-party Communication Models} Yao introduced the two-party communication complexity framework in 1979 \cite{y79}, which was extended to multiple parties by Chandra et al. a few years later \cite{cfl83}. The various multi-party communication complexity models known today can be categorized by 1) how input data is partitioned among the parties; and 2) how the parties communicate with each other. Regarding 1), in the {\em number-on-the-forehead model}, every party knows the input of all other parties except their own, which was also the model considered by Chandra et al. This stands in contrast to the {\em number-in-hand model}, where every party knows their own input, which is the model studied in this paper.
Regarding 2), in the {\em message passing model}, every pair of vertices has access to a private communication link that is not visible to other parties. This model is similar to the {\em coordinator model}, where every party solely communicates with a third party referred to as the {\em coordinator}. The two models are identical up to a $O(\log k)$ factor in the communication complexity (see, for example, \cite{pvz12}). The {\em blackboard model} can be regarded as a variant of the coordinator model, where every message received by the coordinator is immediately visible to every other party. Writing on the blackboard can also be regarded as a message broadcast. %

\subsubsection*{Multi-party Communication Complexity of Graph Problems} Various graph problems have been studied in the multi-party number-in-hand communication settings. Typically, the edges of the input graph are distributed among the participating parties, either with or without duplication, with the without duplication setting usually being easier to handle. In the message passing model, tight bounds are known for testing bipartiteness, cycle-freeness, connectivity, and degree computations \cite{wz17}, approximate maximum matching and approximate maximum flow \cite{hrvz20}, and spanner computations \cite{fwy20}. In contrast, graph problems have not been considered in the blackboard model to the best of our knowledge. However, the Maximum Matching problem has been considered in the {\em simultaneous message model}, which can be seen as a variant of the blackboard model, where all parties simultaneously send a single message to a coordinator who then outputs the result \cite{dno14,anrw15,k15,akly16}.

\subsubsection*{Further Results in the Blackboard Model} While graph problems have not yet been studied in the blackboard model, the model has nevertheless attracted significant attention in recent years. For example,  tight bounds are known for Set-Disjointness \cite{bo15}, bit-wise XOR and logical AND computations \cite{pvz12}, optimization problems such as solving linear systems \cite{vww20}, and distributed task allocation problems \cite{dko12}. %
We point out that the blackboard model also abstracts aspects of real-world graph processing systems with shared memory such as \cite{graphlab} and, as observed in \cite{bo15}, models single-hop wireless ad hoc networks where nodes communicate via broadcast.

\subsection{Techniques}
Consider a graph problem $\textsc{P}$ and the uniform distribution $\lambda$ over all $t$-vertex graphs, for some constant $t$. We denote a graph chosen from $\lambda$ as {\em gadget}. Our hard input graph distribution is the product distribution $\mu = \lambda^{n/t}$, i.e., graphs consisting of $n/t$ disjoint and independent gadgets chosen from $\lambda$.

Let $\mathcal{P}$ be a $\bb(k)$ communication protocol for problem $\textsc{P}$. We measure the amount of information about the edges of the input graph that is necessarily revealed by the transcript of $\mathcal{P}$, i.e., by the bits written on the blackboard. This quantity is usually referred to as the {\em external information cost} (see \cite{b17} for an excellent exposition) of a protocol $\mathcal{P}$ and constitutes a lower bound on the communication cost of $\mathcal{P}$.

At the core of our lower bound argument lies a direct sum argument. Towards a contradiction, assume that the external information cost of $\mathcal{P}$ is $\delta \cdot n$, for a sufficiently small constant $\delta \ll \frac{1}{t}$. We then prove that $\mathcal{P}$ can be used to obtain a protocol $\mathcal{Q}$ with external information cost $t \cdot \delta$ that solves $\textsc{P}$ on a single $t$-vertex gadget, i.e., on inputs chosen from $\lambda$. We further argue that since $t \cdot \delta$ is a very small constant, $\mathcal{Q}$ cannot depend much on the actual input of the problem.
We then give a compression lemma, showing that there is a protocol $\mathcal{Q}'$ that avoids communication altogether, while only marginally increasing the probability of error. Finally, to show that such a ``silent'' protocol $\mathcal{Q}'$ cannot exist, we employ Yao's lemma, which allows us to focus on deterministic protocols with distributional error, and then give simple problem-specific combinatorial arguments.

\subsection{Outline}
In Section~\ref{sec:prelim}, we define the $\bb(k)$ model and provide the necessary context on information theory. %
Next, we present our main result, our lower bound for non-trivial graph problems in the blackboard model, in Section~\ref{sec:lb}. In Section~\ref{sec:upper-bounds}, we summarize how known algorithms can be adapted to yield optimal (up to poly-logarithmic factors) communication protocols in the blackboard model. Last, we conclude in Section~\ref{sec:conclusion}.

\section{Preliminaries and Computing Models} \label{sec:prelim}
\subsection{The Blackboard Model}
In the \emph{(shared) blackboard model}, denoted by $\bb(k)$, we have $k$ parties that communicate by writing messages on a shared blackboard that is visible to all parties. The way the parties interact, and, in particular, the order in which the parties write onto the blackboard, is solely specified by the given communication protocol (and the current state of the blackboard), i.e., the model is therefore asynchronous.
All parties have access to both private and public random coins.

In this paper, we study graph problems in the blackboard model. Given an input graph $G=(V, E)$ with $n = |V|$, we consider an edge-partition model without duplications, where the edges of $G$ are partitioned among the $k$ parties according to a partition function $\prtn: [V \times V] \rightarrow [k]$ that maps the set of {\em all pairs of vertices} of $G$ to the machines (since we consider undirected graphs, we assume that $\prtn(u,v) = \prtn(v,u)$, for every $u,v \in V$). This function does not reveal the set of edges of $G$ itself, however, it reveals that if an edge $uv$ is contained in the input graph, then it is located at machine $\prtn(u,v)$. For our lower bounds, we assume that $\prtn$ is chosen uniformly at random from the set of all symmetric mappings from $[V \times V]$ to $[k]$, and that $\prtn$ is known by all parties. For our upper bounds, we assume that $\prtn$ is chosen adversarially and $\prtn$ is not given to the parties. Instead, each machine $i$ then only knows the edges that are assigned to it at the beginning of the algorithm, i.e., the set of edges $\{uv \in E \ : \ \prtn(u,v) = i \}$.

For a given problem $\textsc{P}$, we say that a \emph{communication protocol $\mathcal{P}$ has error $\epsilon$} if, for any input graph, the probability that the joint output of the parties is a correct solution for $\textsc{P}$ is at least $1-\epsilon$.
This probability is taken over all private and public coins,
and the random selection of the partition function $Z$.\footnote{We adapt this definition accordingly when considering adversarially chosen partition functions for the upper bounds in Section~\ref{sec:upper-bounds}.}
We consider the global output of $\mathcal{P}$ as the union of the local outputs of the machines. For example, for MIS, every party is required to output an independent set so that the union of the parties' outputs constitutes an MIS in the input graph. We say that a protocol is {\em silent} if none of the parties write on the blackboard.

The transcript of a protocol $\mathcal{P}$, denoted by $\Pi(\mathcal{P})$, is the entirety of the bits written onto the blackboard. The \emph{communication cost of a protocol $\mathcal{P}$} is the maximum length of $\Pi(\mathcal{P})$ in an execution of $\mathcal{P}$, which we denote by $|\Pi(\mathcal{P})|$. Then, the \emph{randomized $\epsilon$-error blackboard communication complexity} of a problem $\textsc{P}$, denoted by $\textsc{CC}_\epsilon(P)$ is the minimum cost of a protocol $\mathcal{P}$ that solves $\textsc{P}$ with error at most $\epsilon$.

\subsection{Tools From Information Theory}\label{sec:tools}
Let $A$ be a random variable distributed according to a distribution $\mathcal{D}$. We denote the \emph{(Shannon) entropy} of $A$ by $H_{\mathcal{D}}(A)$, where the index $\mathcal{D}$ may be dropped if the distribution is clear from the context.
The \emph{mutual information} of two random variables $A, B$ distributed according to $\mathcal{D}$ is denoted by $\II_\mathcal{D}[A \ : \ B] = H_\mathcal{D}(A) - H_\mathcal{D}(A \ \mid \ B)$ (again, $\mathcal{D}$ may be dropped), where $H_\mathcal{D}(A \ \mid \ B)$ is the entropy of $A$ conditioned on $B$.
We frequently use the shorthand $\{c\}$ for the event $\{C \!=\! c\}$, for a random variable $C$, and we write $\EE_C$ to denote the expected value operator where the expectation is taken over the range of $C$.
The \emph{conditional mutual information} of $A$ and $B$ conditioned on random variable $C$ is defined  as
\[
\II_{\mathcal{D}}[ A : B \mid C ] = \EE_{C}\lb[ \II_{\mathcal{D}\mid C = c}[ A : B \mid C \!=\! c ]\rb] = \sum_{c} \Pr[ c ] \II_{\mathcal{D}\mid C = c}[ A : B \mid c ].
\]
For a more detailed introduction to information theory, we refer the reader to \cite{ct06}.

We will use the following properties of mutual information:

\begin{enumerate}
    \item \textit{Chain rule for mutual information.} $\II [A \ : \ B, C ] = \II [A \ : \ B ] + \II [A \ : \ C \ \mid \ B]$.
    \item \textit{Independent events.} For random variables $A,B,C$ and event $E$: $\II [A \ : \ B \mid C, E] = \II[A \ : \ B \mid C]$, if  $E$ is independent of $A, B, C$.
    \item \textit{Conditioning on independent variable.} (see e.g. Claim 2.3. in \cite{akl16} for a proof) Let $A,B,C,D$ be jointly distributed random variables so that $A$ and $D$ are independent conditioned on $C$. Then, $\II[A \ : \ B \mid C,D] \ge \II[A \ : \ B \mid C] \ .$
\end{enumerate}
Our lower bound proof relies on Pinsker's Inequality, which relates the total variation distance to the Kullback-Leibler divergence of two distributions:

\begin{definition}[Total Variation Distance, Kullback-Leibler Divergence] \label{def:totalvariation}
  Let $X$ be a discrete random variable and consider two probability distributions $\mathbf{p}(X)$ and $\mathbf{q}(X)$.
  The \emph{total variation distance} between $\mathbf{p}(X)$ and $\mathbf{q}(X)$ is defined as
  \[
    \delta_{TV} \lb( \mathbf{p}(X), \mathbf{q}(X) \rb) = \frac{1}{2} \sum_{x}^{} \lb| \mathbf{p}(x) - \mathbf{q}(x) \rb|,
  \]
  and the \emph{Kullback–Leibler divergence from $\mathbf{q}(X)$ to $\mathbf{p}(X)$ (measured in bits)} is defined as
  \begin{align*}
    \DD \lb[ \mathbf{p}(X) \parallel \mathbf{q}(X) \rb]
    =
    \sum_{x} \mathbf{p}(x) \log \frac{\mathbf{p}(x)}{\mathbf{q}(x)}.
  \end{align*}
\end{definition}

\begin{theorem}[Pinsker's Inequality~\cite{pinsker1964information}] \label{thm:pinsker}
  Let $X$ be a random variable. For any two probability distributions $\mathbf{p}(X)$ and $\mathbf{q}(X)$, it holds that
  \[
    \delta_{TV}\lb(\mathbf{p}(X),\mathbf{q}(X)\rb) \le \sqrt{\frac{\log 2}{2} \DD \lb[ \mathbf{p}(X) \parallel \mathbf{q}(X) \rb]}.
  \]
\end{theorem}

\section{Lower Bound} \label{sec:lb}

\subsection{Input Distributions}
In our input distributions, we assume that the vertex set $V$ of our input graph $G=(V, E)$ is fixed and known to all parties. Furthermore, w.l.o.g. we assume that $V = [|V|]$\footnote{For an integer $\ell$ we write $[\ell]$ to denote the set $\{1, 2, \dots, \ell\}$.}. We will now define a distribution $\rho_N$ of partition functions that map the set of potential edges $V \times V$ to the parties $[k]$, where the subscript $N$ denotes the cardinality of the set of vertices, and two distributions $\lambda$ and $\mu$ over the edge set of our input graph. Since we need to specify both the partition function and the input graph as the input to a protocol in the $\bb(k)$ model, the relevant distributions are therefore the product distributions $\lambda \times \rho_N$ and $\mu \times \rho_N$.

\subsubsection*{Partition Function}
We denote by $\rho_N$ the uniform distribution over all symmetric functions mapping $V \times V$ to $[k]$, where $|V| = N$, and we denote by $\prtn \sim \rho_N$ the partition function associated with our input. In this section, we assume that $\prtn$ is known to all parties.

\subsubsection*{Input Graph}
We consider two input distributions, which we denote the {\em single-gadget} and the  {\em multi-gadget} distributions. For a given {\em gadget size} $t \in O(1)$, let $\lambda$ be the probability distribution over $t$-vertex graphs where we sample each edge uniformly and independently with probability $\tfrac{1}{2}$.
We call the resulting graph a \emph{gadget} and say that $\lambda$ is a \emph{single-gadget distribution}.
We define the \emph{multi-gadget distribution} $\mu$ to be the probability distribution where we sample $n/t$ gadgets independently from $\lambda$. More precisely, we assume that gadget $i$ is induced by vertices $\{(i-1)t + 1, \dots, it \}$, and the edges interconnecting gadget $i$ are distributed according to $\lambda$. For simplicity, we assume that $n/t$ is an integer.
To differentiate between the two distributions, we always denote the number of vertices of the input graph when chosen from the single-gadget distribution $\lambda$ by $t$, and when chosen from the multi-gadget distribution $\mu$ by $n$.

\subsection{Non-trivial Graph Problems} \label{sec:nontrivial}
The lower bounds proved in this paper hold for {\em non-trivial} problems, which are defined by means of {\em good partitions}:
\begin{definition}[Good Partition]
Consider a single-gadget partition function $z$. We say that $z$ is {\em good}, if every party receives at most one potential edge under $z$.
\end{definition}
Observe that this is only possible if the number of parties exceeds the number of potential edges, i.e., if $k \ge {t \choose 2}$. In our setting, $t$ is a small constant and $k$ is sufficiently large. Almost all partitions are therefore good.

\begin{definition}[Non-trivial Problem]
We say that a graph problem is non-trivial if there exists a natural number $t$ such that for every good partition $z$ there is no deterministic silent $\bb(k)$ model protocol that solves the problem on every instance $(G,z)$, where $ G \sim \lambda$.
\end{definition}
For the parties to contribute to the global output of the protocol by only seeing at most one edge and without communication is impossible for most problems. The class of non-trivial problems is therefore large. In Section~\ref{sec:lb-silent}, we give formal proofs, demonstrating that many key problems considered in distributed computing are non-trivial.

\subsection{Information Cost}
We now define the (external) information cost of a protocol in the $\bb(k)$ model:

\begin{definition}[(External) information cost]
Let $\mathcal{P}$ be an $\epsilon$-error protocol in the $\bb(k)$ model for a problem $\textsc{P}$. We define
\[
\ICost_{\mu \times \rho_n}(\mathcal{P}) = \II_{\mu \times \rho_n}[ E(G) : \Pi(\mathcal{P}) \mid \prtn, R_{\mathcal{P}} ],
\]
where $R_{\mathcal{P}}$ are the public random coins used by $\mathcal{P}$ and $\prtn$ is the partition function.
\end{definition}
 Informally, $\ICost_{\mu \times \rho_n}(\mathcal{P})$ is the amount of information revealed about the input edge set by the transcript, given that the public randomness and the partition function are known. We define $\ICost_{\lambda \times \rho_t}(\mathcal{P})$ in a similar way.

It is well-known that the external information cost of a protocol is a lower bound on its communication cost. For completeness, we give a proof in the following lemma:

\begin{lemma} \label{lem:cc}
    $\textsc{CC}_\epsilon(P) \ge \min_{\mathcal{Q}}\ICost_{\mu \times \rho_n}(\mathcal{Q}),$
    where $\mathcal{Q}$ ranges over all $\epsilon$-error randomized protocols for problem $P$.
\end{lemma}
\begin{proof}
\begin{align*}
    \min_{\mathcal{Q}}
       \ICost_{\mu \times \rho_n}(\mathcal{Q})
       & =
    \min_{\mathcal{Q}}
        \II_{\mu \times \rho_n}[ E(G) : \Pi(\mathcal{P}) \mid \prtn, R_{\mathcal{Q}} ]
       \le
   \min_{\mathcal{Q}}
       \HH_{\mu \times \rho_n}[ \Pi(\mathcal{Q}) \mid \prtn, R_{\mathcal{Q}} ]
       \le
   \min_{\mathcal{Q}}
       \HH_{\mu \times \rho_n}[ \Pi(\mathcal{Q})  ]
       \le
   \min_{\mathcal{Q}}
        |\Pi(\mathcal{Q})|
       \\
       & \le
       \textsc{CC}_\epsilon({P}),
\end{align*}
  where the second last step follows from Theorem~6.1 in \cite{ccbook}.
\end{proof}

\subsection{Direct Sum Argument}
In this section, we show that the information cost of a protocol for the multi-gadget case is lower-bounded by the information cost of a protocol for the single-gadget case, multiplied by the number of gadgets:
\begin{theorem}\label{th:P-Qprotocols}
Let $\mathcal{P}$ be an $\epsilon$-error randomized $\bb(k)$ protocol. Then, there exists an $\epsilon$-error $\bb(k)$ protocol $\mathcal{Q}$ such that:
\begin{align}
    \ICost_{\lambda \times \rho_t}(\mathcal{Q}) \le \frac{t}{n} \ICost_{\mu \times \rho_n}(\mathcal{P})  \ .
    \label{eq:equiv_mutual}
  \end{align}
\end{theorem}
\begin{proof}
Given $\mathcal{P}$, we construct protocol $\mathcal{Q}$ as follows. Let $(H,\prtn) \sim \lambda \times \rho_t$ denote the input to $\mathcal{Q}$. The parties construct input $(G,\prtn')$ for $\mathcal{P}$ from $(H,\prtn)$, as follows:

\begin{enumerate}
    \item The parties use public randomness $R_{\mathcal{Q}}$ to sample a uniform random index $I \in [n/t]$.
    \item The parties use public randomness to sample a partition $\prtn'$ from $\rho_n$ conditioned on the partition of gadget $I$ being equivalent to $\prtn$.
    \item The parties set the single input gadget $H$ to be gadget $I$ in the multi-gadget input $G$.
    \item Given $\prtn'$, the parties know which potential edges of the gadgets different to gadget $I$ are hosted locally. They then use private randomness to sample the existence of the individual edges.
    \item The parties run protocol $\mathcal{P}$ on input $(G, \prtn')$ and return the part of the output that concerns gadget $I$.
\end{enumerate}
Observe that the random variables $R_{\mathcal{Q}}, \prtn$ are identical to $I, R_{\mathcal{P}}, \prtn'$, where $R_{\mathcal{P}}$ denotes the public coins used by protocol $\mathcal{P}$. Observe further that the constructed input $(G,\prtn')$ is distributed according to $\mu \times \rho_n$.

Denote by $E_I$ the edges of gadget $I$ of the multi-gadget input. We obtain:

\begin{align}
     \ICost_{\lambda \times \rho_t}(\mathcal{Q}) & =  \II_{\lambda \times \rho_t}[E(H) : \Pi(\mathcal{Q}) \mid R_{\mathcal{Q}}, \prtn]  \nonumber \\
     & = \II_{\mu \times \rho_n}[E_I : \Pi(\mathcal{P}) \mid R_{\mathcal{P}}, I, \prtn'] \nonumber \\
     & = \EE_{i \sim I} \lb[\II_{\mu \times \rho_n}[E_i : \Pi(\mathcal{P}) \mid R_{\mathcal{P}}, \prtn', I = i]  \rb] \nonumber \\
     & \le \frac{t}{n}\sum_{i \in [n/t]} \II_{\mu \times \rho_n}[E_i : \Pi(\mathcal{P}) \mid R_{\mathcal{P}}, E_1, \dots, E_{i-1}, \prtn', I = i] \label{eqn:921},
  \intertext{
  where (\ref{eqn:921}) follows from Property~(3) in Section~\ref{sec:tools} since $E_i$ and $E_1,\dots,E_i$ are independent.
  Moreover, we observe that $E_i$ and $\Pi(\mathcal{P})$ are independent of the event $\{ I \!=\! i \}$, which yields
  }
  \ICost_{\lambda \times \rho_t}(\mathcal{Q})
     & \le \frac{t}{n}\sum_{i \in [n/t]} \II_{\mu \times \rho_n}[E_i : \Pi(\mathcal{P}) \mid R_{\mathcal{P}}, E_1, \dots, E_{i-1}, \prtn'] \nonumber \\
     & \le \frac{t}{n} \II_{\mu \times \rho_n}[E(G) : \Pi(\mathcal{P}) \mid R_{\mathcal{P}}, \prtn'] \label{eqn:922} \\
     &= \frac{t}{n}\ICost_{\mu\times\rho_n}(\mathcal{P}).\nonumber
\end{align}
where we used the chain rule of mutual information (Property~(1) in Section~\ref{sec:tools}) in \eqref{eqn:922}.
\end{proof}

\subsection{Zero Communication Protocol}
Next, we show that if a protocol $\mathcal{Q}$ solves single-gadget instances with small information cost, then there exists a silent protocol (i.e., a protocol that avoids all communication) with only slightly increased error.

\begin{theorem} \label{thm:zero}
For a problem $\textsc{P}$, let $\mathcal{Q}$ be an $\epsilon$-error randomized protocol in the $\bb(k)$ model with
\begin{align}
\ICost_{\lambda \times \rho_t}(\mathcal{Q}) \le \lb(\tfrac{2\epsilon}{3}\rb)^4 \ . \label{eq:inf_cost_upper_bnd}
\end{align}
Then, there exists a silent protocol in the $\bb(k)$ model with error at most $4\epsilon$ for problem $\textsc{P}$.
\end{theorem}
\begin{proof}
Consider the protocol $\mathcal{Q}$ as stated in the premise of the theorem.
As a first step, we leverage the assumption that $\ICost_{\lambda \times \rho_t}(\mathcal{Q})$ is small to show that the messages sent by the algorithm are close to being independent of the edges of the sampled gadget $H$.

\begin{lemma} \label{lem:total_variation}
  Suppose we sample the edges $E(H) \sim \lambda$ of gadget $H$ and a partition function $Z \sim \rho_t$. Let $R$
  be the public randomness used by $\mathcal{Q}$.
  Let $\mathbf{p}(\Pi(\mathcal{Q}) \mid E(H), R, \prtn)$ denote the probability distribution of $\Pi(\mathcal{Q})$ conditioned on $E(H)$, $R$, and $\prtn$, and define $\mathbf{p}(\Pi(\mathcal{Q}) \mid R, \prtn )$ similarly.
  Then, with probability at least $1 - \tfrac{2\epsilon}{3}$, it holds that
  \begin{align}
    \delta_{TV}\Big(
      \mathbf{p}(\Pi(\mathcal{Q}) \mid E(H), R, \prtn),
      \mathbf{p}(\Pi(\mathcal{Q}) \mid R, \prtn )
    \Big)
    \le \tfrac{2\epsilon}{3} \ .
    \label{eq:total_variation}
  \end{align}
\end{lemma}
\begin{proof}
  We will derive an upper bound on the expected total variation distance, where the expectation ranges over the input gadget $H$, the public random string $R$, and the partition function $Z$. In the following, we use the abbreviation $E=E(H)$.
  By applying Pinsker's inequality (see Theorem~\ref{thm:pinsker}) and taking expectations on both sides, we get
  \begin{align}
    \EE_{E,R,Z}\lb[
        \delta_{TV}\lb(
          \mathbf{p}(\Pi(\mathcal{Q}) \mid E,R,Z),
          \mathbf{p}(\Pi(\mathcal{Q}) \mid R,Z)
        \rb)
    \rb]
    &\le
    \EE_{E,R,Z}\lb[
        \sqrt{\frac{\log 2}{2}
          \DD \lb[
              \mathbf{p}(\Pi(\mathcal{Q}) \mid E,R,Z)
              \parallel
              \mathbf{p}(\Pi(\mathcal{Q}) \mid R,Z)
          \rb]}
    \rb] \notag\\
  \intertext{
    From Jensen's inequality for concave functions and the fact that $\frac{\log 2}{2} \le 1$, we get
  }
  \EE_{E,R,Z}\lb[
      \delta_{TV}\lb(
        \mathbf{p}(\Pi(\mathcal{Q}) \mid E,R,Z),
        \mathbf{p}(\Pi(\mathcal{Q}) \mid R,Z)
      \rb)
  \rb]
    &\le
    \sqrt{
      \EE_{E,R,Z}\lb[
        \DD \lb[
            \mathbf{p}(\Pi(\mathcal{Q}) \mid E,R,Z)
            \parallel
            \mathbf{p}(\Pi(\mathcal{Q}) \mid R,Z)
        \rb]
      \rb].
    }\notag
  \end{align}
  In the following derivation, $\pi$ ranges over all possible transcripts of $\mathcal{Q}$. By using the definition of the Kullback-Leibler divergence (see Def.~\ref{def:totalvariation}), we obtain that
  \begin{align}
    \EE_{E,R,Z}\lb[
        \delta_{TV}\lb(
          \mathbf{p}(\Pi(\mathcal{Q}) \mid E,R,Z),
          \mathbf{p}(\Pi(\mathcal{Q}) \mid R,Z)
        \rb)
    \rb]
    &\le
      \sqrt{
        \sum_{e,r,z} \Pr\lb[ e,r,z \rb] \sum_{\pi} \Pr\lb[ \pi \mid e,r,z \rb]
                   \log \lb( \frac{\Pr\lb[ \pi \mid e,r,z \rb]}{\Pr\lb[ \pi \mid r,z \rb]} \rb)
      }\notag \\
    &=
      \sqrt{
        \sum_{e,\pi,r,z} \Pr\lb[ e, \pi, r,z \rb]
          \log \lb( \frac{\Pr\lb[ \pi \mid e,r,z \rb] \Pr\lb[ e \mid r,z \rb] \Pr\lb[ r,z \rb]}{\Pr\lb[ e \mid r,z \rb] \Pr\lb[ \pi \mid r,z \rb] \Pr\lb[ r,z \rb]} \rb)
      }\notag\\
    &=
      \sqrt{
        \sum_{e,\pi,r,z} \Pr\lb[ e, \pi, r, z \rb] \log \lb( \frac{\Pr\lb[ e, \pi, r,z \rb]}{\Pr\lb[ e \mid r,z \rb] \Pr\lb[ \pi \mid r,z \rb] \Pr\lb[ r,z \rb]} \rb)
      }\notag\\
    &= \sqrt{ \II[ E(H) : \Pi(\mathcal{Q}) \mid R,Z ]} \label{eq:mutual1} \\
    &= \sqrt{ \ICost_{\lambda\times\rho_t}(\mathcal{Q})} \notag \\
    &\le \lb(\tfrac{2\epsilon}{3}\rb)^2, \notag
  \end{align}
  where we used the definition of conditional mutual information in \eqref{eq:mutual1}, and the upper bound \eqref{eq:inf_cost_upper_bnd} in the last step.
  To translate the upper bound on the expected value to an upper bound on the probability that the total variation distance is large, we apply Markov's inequality.
  That is, if we sample $E$, $R$, and $Z$ from the respective distributions stated in the premise of the lemma, then it holds that
  \begin{align}
    \Pr\lb[
      \delta_{TV}\lb( \mathbf{p}(\Pi(\mathcal{Q}) \mid E,R,Z), \mathbf{p}(\Pi(\mathcal{Q}) \mid R,Z) \rb)
      \ge
      \tfrac{2\epsilon}{3}
    \rb]
    &\le
    \frac{3}{2\epsilon}
      \EE_{E,R,Z}\lb[
        \delta_{TV}\lb( \mathbf{p}(\Pi(\mathcal{Q}) \mid E,R,Z), \mathbf{p}(\Pi(\mathcal{Q}) \mid R,Z) \rb)
      \rb]
     \notag\\
    &\le \frac{2\epsilon}{3},\notag
  \end{align}
  and the lemma follows.
\end{proof}

Equipped with Lemma~\ref{lem:total_variation}, we are now ready to prove Theorem~\ref{thm:zero}.
The players use the even bits of their public randomness to sample $R$, which is the public randomness of protocol $\mathcal{Q}$, yielding a new protocol $\mathcal{Q}_R$ that does not use public randomness but otherwise behaves the same way as $\mathcal{Q}$ given $R$.
Then, the players use the odd bits of their public randomness to sample a transcript $\Pi$ from the distribution $\mathbf{p}(\Pi(\mathcal{Q}_R) \mid Z) = \mathbf{p}(\Pi(\mathcal{Q}) \mid R,Z)$.
After these two steps, all players know $R$ and $\Pi$.
Finally, each player $P_i$ computes its output by applying the state transition function of $\mathcal{Q}_R$ to  $Z$, $\Pi$, its input assignment, and its private random bits.
Thus, we have obtained a silent protocol $\mathcal{S}$.

To prove the lemma, we need to obtain an upper-bound on the error probability of $\mathcal{S}$.
Clearly, $\mathcal{S}$ fails in all instances where $\mathcal{Q}$ fails.
In fact, the only difference between computing the output of $\mathcal{S}$ compared to $\mathcal{Q}$ is that the transcript used in $Q$ is sampled from distribution $\mathbf{p}(\Pi(\mathcal{Q}) \mid E(H),R,Z)$, whereas the one we use in $\mathcal{S}$ is sampled from $\mathbf{p}(\Pi(\mathcal{Q}) \mid R,Z)$.
Thus, the only additional error probability of protocol $\mathcal{S}$ is determined by the difference in the probability mass assigned to each transcript $\pi$ between the distributions $\mathbf{p}(\Pi(\mathcal{Q}) \mid E(H),R,Z)$ and $\mathbf{p}(\Pi(\mathcal{Q}) \mid R,Z)$.
We obtain
\begin{align}
  \Pr\lb[ \text{$\mathcal{S}$ errs} \rb]
    &=     \Pr\lb[ \text{$\mathcal{S}$ errs} \mid \text{$\mathcal{Q}$ errs }\rb]
           \Pr\lb[ \text{$\mathcal{Q}$ errs} \rb]
         + \Pr\lb[ \text{$\mathcal{S}$ errs} \mid \text{$\mathcal{Q}$ correct}\rb]
           \Pr\lb[ \text{$\mathcal{Q}$ correct} \rb] \notag\\
    &\le \epsilon
         + \Pr\lb[ \text{$\mathcal{S}$ errs} \mid \text{$\mathcal{Q}$ correct}\rb].
         \label{eq:probbound1}
\end{align}
Let $\set{\delta_{TV} \le \tfrac{2\epsilon}{3}}$ be the event that
  $\delta_{TV}\Big( \mathbf{p}(\Pi(\mathcal{Q}) \mid E(H),R, Z),
                    \mathbf{p}(\Pi(\mathcal{Q}) \mid  R, Z) \Big)
    \le \tfrac{2\epsilon}{3}$,
and define event $\set{\delta_{TV} > \tfrac{2\epsilon}{3}}$ similarly.
To upper-bound the probability $\Pr\lb[ \text{$\mathcal{S}$ errs} \mid \text{$\mathcal{Q}$ correct}\rb]$, we condition on whether event $\set{\delta_{TV} \le \tfrac{2\epsilon}{3}}$ holds.
Continuing from \eqref{eq:probbound1}, we get
\begin{align}
  \Pr\lb[ \text{$\mathcal{S}$ errs} \rb]
    &\le \epsilon
        + \Pr\lb[ \text{$\mathcal{S}$ errs}
             \mid \delta_{TV} \le \tfrac{2\epsilon}{3}, \text{$\mathcal{Q}$ correct}
          \rb]
          \Pr\lb[ \delta_{TV} \le \tfrac{2\epsilon}{3}\
             \middle|\ \text{$\mathcal{Q}$ correct}
          \rb] \notag\\
    &\phantom{-----}
        + \Pr\lb[ \text{$\mathcal{S}$ errs}
             \mid \delta_{TV} > \tfrac{2\epsilon}{3}, \text{$\mathcal{Q}$ correct}
          \rb]
          \Pr\lb[ \delta_{TV} > \tfrac{2\epsilon}{3}\
             \middle|\ \text{$\mathcal{Q}$ correct}
          \rb] \notag\\
    &\le \epsilon
        + \Pr\lb[ \text{$\mathcal{S}$ errs}
             \mid \delta_{TV} \le \tfrac{2\epsilon}{3}, \text{$\mathcal{Q}$ correct}
          \rb]
        + \Pr\lb[ \delta_{TV} > \tfrac{2\epsilon}{3}\
           \middle|\ \text{$\mathcal{Q}$ correct}
        \rb]. \label{eq:probbound2}
\end{align}
Conditioned on $\mathcal{Q}$ being correct, we know that the additional error of $\mathcal{S}$ is at most
\[
  \sum_\pi \big| \Pr\lb[ \Pi(\mathcal{Q}) \!=\! \pi \mid E(H),R,Z \rb]
                - \Pr\lb[ \Pi(\mathcal{Q}) \!=\! \pi \mid R,Z \big]
           \rb| \le
           2\delta_{TV}\lb(\mathbf{p}(\Pi(\mathcal{Q}) \mid E(H),R,Z),\mathbf{p}(\Pi(\mathcal{Q}) \mid R,Z)\rb),
\]
which, conditioned on event $\set{\delta_{TV} \le \tfrac{2\epsilon}{3}}$, yields the bound
$\Pr\lb[ \text{$\mathcal{S}$ errs} \mid \delta_{TV} \le \tfrac{2\epsilon}{3}, \text{$\mathcal{Q}$ correct} \rb] \le \tfrac{4\epsilon}{3}$.

To complete the proof, we will derive an upper bound on
$\Pr\lb[ \delta_{TV} > \tfrac{2\epsilon}{3}\ \middle|\ \text{$\mathcal{Q}$ correct} \rb]$.
Observe that
\begin{align}
  \Pr\lb[ \delta_{TV} \le \tfrac{2\epsilon}{3} \rb]
  &\le \Pr\lb[ \delta_{TV} \le \tfrac{2\epsilon}{3} \mid \text{$\mathcal{Q}$ correct}\rb]
     + \Pr\lb[ \text{$\mathcal{Q}$ errs} \rb] \notag\\
  &\le
     \Pr\lb[ \delta_{TV} \le \tfrac{2\epsilon}{3} \mid \text{$\mathcal{Q}$ correct}\rb]
     + \epsilon.\notag
\end{align}
By applying Lemma~\ref{lem:total_variation} to the left-hand side, we get
\begin{align}
  \Pr\lb[ \delta_{TV} > \tfrac{2\epsilon}{3} \mid \text{$\mathcal{Q}$ correct}\rb]
  &\le \tfrac{2\epsilon}{3} + \epsilon.\notag
\end{align}
Plugging these two bounds into the right-hand side of \eqref{eq:probbound2}, we get
$
  \Pr\lb[ \text{$\mathcal{S}$ errs} \rb]
    \le 2\epsilon + \frac{6\epsilon}{3} = 4\epsilon.
$
This completes the proof of Theorem~\ref{thm:zero}.
\end{proof}

\subsection{From Randomized to Deterministic Protocols}

In this section, we show that the existence of
a randomized silent protocol for single-gadget instances with a sufficiently small error implies existence of a deterministic silent protocol for single-gadget instances that solves the problem on every input graph under some good partition.

First, we will argue that the probability that the partition function is good is at least $1/2$:

\begin{lemma}\label{lem:good-partitionings}
 Let $t \geq 1$ and $k \geq 3t^4$.
 Consider a single gadget input on $t$ vertices. Then, the probability that the partition function is good is at least $1/2$.
\end{lemma}
\begin{proof}
There are ${t \choose 2} \le t^2$ vertex pairs. The probability that they are all assigned to different parties is at least:
$$\frac{k \cdot (k-1) \cdot (k-2) \cdot \dots \cdot (k-t^2+1)}{k^{t^2}} \ge \left( \frac{k-t^2+1}{k} \right)^{t^2} = \lb(1 - \frac{t^2-1}{k}\rb)^{t^2} \ge \exp \left( \frac{-2t^2(t^2 - 1)}{k}\right)  \ge \frac{1}{2}\ ,$$
where in the penultimate inequality we used the fact that $1-x \ge \exp(-2x)$ for any $x \in [0, 1/2]$.
\end{proof}

We now state and prove the main result of the section.

\begin{lemma}\label{lem:derand}
    Let $t \geq 1$, $k \geq 3t^4$, and $\epsilon \le \frac{1}{3 {t \choose 2} }$.
    Let $\mathcal{P}$ be a randomized silent protocol with error at most $\epsilon$ (where the error is over the random partition $\rho_t$ and the random coin flips).
    Then, there exists a good partition $z$ and a deterministic silent protocol that succeeds on every input $(H, z)$, where $H$ is a $t$-vertex graph.
\end{lemma}
\begin{proof}
    First, the assumption of the lemma together with Yao's minimax principle \cite{yao1977probabilistic} imply that there exists a deterministic silent protocol $\mathcal{Q}$ with distributional error at most $\epsilon$ over the input distribution $\lambda \times \rho_t$.

    Now, for a fixed partition function $z$, observe that if $\mathcal{Q}$ is not correct on {\em all} $t$-vertex input graphs, then its error conditioned on $z$ is at least $\frac{1}{{t \choose 2}}$. For the sake of a contradiction, suppose that there is no good partition such that $\mathcal{Q}$ succeeds on all inputs. Recall further that at least half of all partition functions are good ( Lemma~\ref{lem:good-partitionings}). Then:

    \begin{align*}
        \epsilon = \Pr [\mathcal{Q} \text{ errs}] & \ge  \Pr [z \text{ good} ] \cdot \Pr [\mathcal{Q} \text{ errs} \mid z \text{ good}]  \ge \frac{1}{2} \frac{1}{{t \choose 2}} \ ,
    \end{align*}
    contradicting the assumption that $\epsilon \le \frac{1}{3} \frac{1}{{t \choose 2}}$.
\end{proof}

\subsection{Main Lower Bound Theorem}
\textbf{Theorem \ref{thm:main}.}
\textit{
$\textsc{P}$ be a non-trivial graph problem. Then there exists $\delta(P) > 0$ and $t = t(P)$ such that
for any $k \geq 3t^4$ and $\epsilon \leq \delta(P)$ the randomized $\epsilon$-error blackboard communication complexity $\textsc{CC}_\epsilon(P)$ of $P$ is $\Omega(n)$, even if the edges of the input graph are partitioned randomly among the parties.
}
\begin{proof}
Let $t = t(P)$ be the minimum integer such that
for every good partition $z$ from $\rho_t$ there is no deterministic silent $\bb(k)$ model protocol that solves $P$ on every instance $(H,z)$, where $ H \sim \lambda$ (such a $t$ exists since $\textsc{P}$ is non-trivial).
Let $\delta(P) = \frac{1}{12{t \choose 2}}$ and $\epsilon \leq \delta(P)$.

Suppose, towards a contradiction, that %
$\textsc{CC}_\epsilon(P) \leq \frac{1}{t}  \left( \frac{2 \epsilon}{3} \right)^4 n$.
Then, by Lemma \ref{lem:cc}, there exists a protocol $\mathcal{P}$ for $P$ with
$\ICost_{\mu \times \rho_n}(\mathcal{P}) \leq \frac{1}{t} \left( \frac{2 \epsilon}{3} \right)^4 n$.
By Theorem \ref{th:P-Qprotocols}, this in turn implies the existence of a protocol
$\mathcal{Q}$ for input instances from $\lambda \times \rho_t$ such that
$$
    \ICost_{\lambda \times \rho_t}(\mathcal{Q}) \leq
    \frac{t}{n}\ICost_{\mu \times \rho_n}(\mathcal{P}) \leq
    \left( \frac{2 \epsilon}{3} \right)^4.
$$

\noindent
Therefore, by Theorem \ref{thm:zero}, there exists a silent
protocol $\mathcal{R}$ for input instances $\lambda \times \rho_t$
with error at most $4\epsilon \leq \frac{1}{3{t \choose 2}}$.
Consequently, by Lemma \ref{lem:derand}, there exists a good partition $z$ from $\rho_t$ and a deterministic silent protocol that succeeds on every input $(G, z)$, where $G$ is a $t$-vertex graph. But this contradicts the choice of $t$.
\end{proof}

\subsection{Lower Bounds for Specific Problems} \label{sec:lb-silent}

In this section we establish non-triviality of several fundamental graph problems.
We point out that lower bounds are known for all of these problems in other distributed computing models such as the \textsc{LOCAL} model (see  \cite{DBLP:conf/focs/Balliu0HORS19,kmw16}).
However, these lower bounds do not directly translate to our setting since, in the $\bb(k)$ model, a player who receives an edge as input does not necessarily need to produce an output for this particular edge or its vertices.

\begin{lemma}
    Maximal Matching is a non-trivial problem.
\end{lemma}
\begin{proof}
For contradiction, suppose that for every natural $t$ there exist a good partition $z$ from $\rho_t$ and a protocol $\mathcal{Q}$ that solves Maximal Matching on every instance chosen from $\lambda$ under partition $z$. Let $t \geq 3$. First, observe that if the input graph is the empty graph, then no party can output an edge. This implies that if a party does not receive an edge then its output needs to be empty. Suppose next that the input graph consists of a single edge. Then, the party that hosts this edge needs to output this edge since otherwise the computed matching would not be maximal. This implies further that any party that holds an edge needs to output their edge. However, if two machines hold edges that share an endpoint then the output would not be a matching, a contradiction.
\end{proof}

\begin{lemma}\label{lem:maxIndSet}
    Maximal Independent Set,
    Minimum Dominating Set,
    $(\Delta+1)$-coloring, and
    Minimum Vertex Cover are non-trivial problems.
\end{lemma}
\begin{proof}
    \textbf{Maximal Independent Set} and \textbf{Minimum Dominating Set}.
    For contradiction, suppose that for every natural $t$ there exist a good partition $z$ from $\rho_t$ and a protocol $\mathcal{Q}$ that solves Maximal Independent Set (resp. Minimum Dominating Set) on every instance chosen from $\lambda$ under partition $z$.
    Let $t \geq 4$. First, observe that if the input graph is the empty graph, all vertices must be output.
    This implies that there is a function $f: [t] \rightarrow [k]$ such that for every $i \in [t]$, if party $f(i)$ hosts no edges, then its output contains vertex $i$.
    Let $a,b,c,d$ be distinct vertices of the input graph. Since there are six different pairs of distinct vertices from $\{a,b,c,d\}$, at least one of these pairs, say $(a,b)$,
    is assigned by the partition $z$ to a party $q \not\in \{ f(a), f(b), f(c), f(d) \}$. But then on the input graph with a unique edge $(a,b)$, the output will contain both vertices $a$ and $b$, which is not possible, as the output of $\mathcal{Q}$ should be an independent set (resp. a \emph{minimum} dominating set).

    \smallskip
    \noindent
    \textbf{$(\Delta+1)$-coloring}. The proof works similarly to the previous one. Indeed, for the empty
    input graph, the output should be a vertex coloring, where every vertex is assigned the same color, say color 1.
    This implies that there is a function $f: [t] \rightarrow [k]$ such that for every $i \in [t]$, if party $f(i)$ hosts no edges, then its output assigns color 1 to vertex $i$.
    Then, following the same arguments as in the previous proof, we conclude that
    on an input graph with at least four vertices $a,b,c,d$ and a unique edge $(a,b)$, the output coloring assigns color 1 to both vertices $a$ and $b$, which is not a proper coloring, a contradiction.

    \smallskip
    \noindent
    \textbf{Minimum Vertex Cover}. Since for any graph $G=(V,E)$ a vertex set $S \subseteq V$ is a minimum vertex cover if and only if the set $V \setminus S$ is a maximum independent set, the existence of a deterministic silent protocol for Minimum Vertex Cover would imply the existence of one for
    Maximum Independent Set. Furthermore, as any \emph{maximum} independent set is also a \emph{maximal} independent set, the existence of a desired protocol for Maximum Independent Set would imply the existence of one for Maximal Independent Set, a contradiction.
\end{proof}

In the proof of non-triviality of Maximal Matching we used gadgets of size 3 and in the proofs of non-triviality of all other problems we used gadgets of size 4. Together with Theorem \ref{thm:main} this implies the following

\begin{corollary}
    \begin{enumerate}
        \item For any $\epsilon \leq \frac{1}{36}$ and $k \geq 243$,
        the randomized $\epsilon$-error blackboard communication complexity of Maximal Matching is $\Omega(n)$.

        \item For any $\epsilon \leq \frac{1}{72}$ and $k \geq 768$,
        the randomized $\epsilon$-error blackboard communication complexity of each of the problems Maximal Independent Set,
        Minimum Dominating Set,
        $(\Delta+1)$-coloring, and
        Minimum Vertex Cover is $\Omega(n)$.
    \end{enumerate}
\end{corollary}

\section{Upper Bounds} \label{sec:upper-bounds}
In this section, we assume that the partition function $z$ is arbitrary (potentially adversarially chosen), and each party $i$ only knows the edges allocated to them, i.e., the set of edges $E_i = \{ uv \in E \ : \ z(u,v) = i \}$.
Thus, in contrast to our lower bounds, the error probability of a protocol is now only computed over the private and public coin flips.

\subsection{Maximal Matching and Maximum Matching Approximation}
Computing a maximal matching in the \bb(k) model with communication cost $O(n \log n)$ is straightforward. The parties can simulate the \textsc{Greedy} matching algorithm, as follows: The blackboard serves as a variable $M$ that contains the current matching. Initially, the blackboard is empty and represents the empty set, i.e., $M \gets \varnothing$. The parties proceed in order. At party $i$'s turn, party $i$ attempts to add each of its edges $e \in E_i$ in any order to the blackboard if possible, i.e., if $M \cup \{e \}$ is still a matching. Since any matching in an $n$-vertex graph is of size at most $n/2$, and each edge requires space $O(\log n)$ to be written onto the blackboard (e.g., by indicating the two endpoints of the edge), the communication cost is $O(n \log n)$.

Various algorithms are known for computing a $(1-\epsilon)$-approximation to Maximum Matching by repeatedly computing maximal matchings in subgraphs of the input graph (e.g. \cite{m05,ag11,ekms12,k18,kk20,alt21}). In order to implement these algorithms in the \bb(k) model, we need to make sure that the $k$ parties know which of their edges are contained in each of the respective subgraphs. If this can be achieved with communication cost $O(s)$ per subgraph, then a \bb(k) model implementation with cost $O((n \log n + s) \cdot r)$ can be obtained, where $r$ is the number of matchings computed.

The algorithm by McGregor \cite{m05} follows this scheme and relies on the computation of  $\frac{1}{\epsilon}^{O(\frac{1}{\epsilon})}$ maximal matchings in order to establish a $(1-\epsilon)$-approximate matching in general graphs. Specifying the respective subgraphs with small communication cost is straightforward for this algorithm, requiring a cost of $O(n \log n)$ (details omitted), and we thus obtain a $(1-\epsilon)$-approximation \bb(k) model algorithm with communication cost $O\left(n \log n \cdot \frac{1}{\epsilon}^{O(\frac{1}{\epsilon})}\right)$.

The dependency on $\epsilon$ can be improved in the case of bipartite graphs. Assadi et al. \cite{alt21} very recently gave an algorithm for Maximum Matching in bipartite graphs that solely relies on the computation of $O(\frac{1}{\epsilon^2})$ maximal matchings in subgraphs of the input graph.
Specifying the respective subgraphs, however, is slightly more involved, but can be done with communication cost $O\left(n \frac{1}{\epsilon} \log(\frac{1}{\epsilon})\right)$. Taking this into account, a \bb(k) model protocol with communication cost $O\left( \left(n \log n + n \frac{1}{\epsilon} \log(\frac{1}{\epsilon}) \right) \cdot \frac{1}{\epsilon^2}\right)$ can be obtained.

\subsection{Maximal Independent Set}
We now discuss how the random order \textsc{Greedy} MIS algorithm can be implemented in the \bb(k) model with communication cost $O(n \poly \log n)$. A similar implementation in the massively parallel computation and the congested clique models was previously discussed by Ghaffari et al. \cite{ggkmr18}. The key idea of this implementation, i.e.,  a residual sparsity property of the random order \textsc{Greedy} algorithm, goes back to the multi-pass correlation clustering streaming algorithm by Ahn et al. \cite{acgmw15} and have since been used at multiple occasions \cite{k18,gkms19}.

The random order \textsc{Greedy} MIS algorithm proceeds as follows: First, the algorithm identifies a uniform random permutation $\pi: V \rightarrow [n]$, assigning each vertex a rank. Denote by $v_i$ the vertex with rank $i$. The algorithm processes the vertices $V$ in the order specified by $\pi$. When processing vertex $v_i$, the algorithm attempts to insert $v_i$ into an initially empty independent set $I$ if possible, i.e., if $I \cup \{v_i \}$ is an independent set.

To simulate this algorithm in the \bb(k) model, the $k$ parties first agree on a uniform random permutation $\pi$ using public randomness. The simulation then proceeds in $O(\log \log n)$ phases, where in phase $i$ \textsc{Greedy} is simulated on vertices $V_i := \{v_j \ : \ k_{i-1} \le j < k_i \}$, for
$$k_i = n^{1-\frac{1}{2^i}} \ .$$
Denote by $I_i$ the independent set computed after phase $i$ (i.e., after having processed all vertices with ranks smaller than $k_i$), and let $I_0 = \{ \}$. In each phase $i$, the parties write the subgraph
induced by vertices $V_i \setminus N(I_{i-1})$ onto the blackboard, where $N(I_{i-1})$ denotes the neighborhood of $I_{i-1}$. This information, together with $I_{i-1}$, then allows all parties to locally continue the simulation of \textsc{Greedy} on vertices $V_i$, resulting in independent set $I_i$. As proved by Ghaffari et al., the subgraph induced by vertices $V_i \setminus N(I_{i-1})$ has $O(n \poly \log n)$ edges with high probability. Since the parties know $I_{i-1}$ from the local simulations of the previous phase, they know which of their edges are contained in subgraph $G[V_i \setminus N(I_{i-1})]$ and write all of these edges onto the blackboard.

After $O(\log \log n)$ phases, it can be seen that the graph induced by vertices with rank larger than $k_{O(\log \log n)}$ has $O(n \poly \log n)$ edges. The parties then write this subgraph onto the blackboard, which allows all parties to complete their local simulations of \textsc{Greedy}.

The communication cost is dominated by the vertex-induced subgraphs written onto the blackboard. Each of these subgraphs contains $O(n \poly \log n)$ edges, and, since there are $O(\log \log n)$ phases, the total communication cost is $O(n \poly \log n)$.

\subsection{$(\Delta+1)$-coloring}
We now discuss how the one-pass $O(n \log^3 n)$-space streaming algorithm by Assadi et al. \cite{ack19} for $(\Delta+1)$-coloring can be implemented in the $\bb(k)$ model.

Assadi et al.'s streaming algorithm assumes that $\Delta$ is known prior to the processing of the stream. For each vertex $v \in V$, the algorithm first samples uniform random $O(\log n)$ colors $P_v \subseteq [ \Delta + 1]$ from the color palette $[\Delta+1]$. Then, while processing the stream, the algorithm maintains the set of edges $uv \in E$ such that $P_u \cap P_v \neq \emptyset$, i.e., the palettes of $u$ and $v$ intersect. Denote this set of edges by $E_{\text{conflict}}$. Assadi et al. prove that $|E_{\text{conflict}}| = O(n \log^2 n)$ w.h.p., and there exists a coloring $\chi: V \rightarrow [\Delta+1]$ of the conflict graph $G_{\text{conflict}} = (V, E_{\text{conflict}})$ such that every vertex $v$ is assigned a color from $P_v$ w.h.p., i.e., $\chi(v) \in P_v$. It is easy to see that this coloring is also a valid coloring of the input graph $G$. Since $|E_{\text{conflict}}| = O(n \log^2 n)$, and space $O(\log n)$ is accounted for storing an edge, the space complexity of this algorithm is $O(n \log^3 n)$.

We implement this algorithm in the \bb(k) model as follows. The parties maintain a guess $\Delta'$ of $\Delta$ that is initially set to $\lceil (n-1)/2 \rceil$ and updated according to a binary search in the range $[n-1]$. In each step of the binary search, our algorithm simulates Assadi et al.'s algorithm using the guess $\Delta'$ instead of $\Delta$. To this end, first, the parties use public randomness to determine the palettes $P_v$ for every vertex $v$. Knowing all the vertices' palettes, every party $i$ is able to identify which of their edges are contained in $E_{\text{conflict}}$. Denote the conflict edges of party $i$ by $E_{\text{conflict}}^i$. Then, every party writes $E_{\text{conflict}}^i$ to the blackboard. Next, the parties locally each compute the same coloring $\chi_{\Delta'}: V \rightarrow [\Delta' + 1]$ of the conflict graph respecting the color palettes of the vertices, if such a coloring exists. If such a coloring exists, the parties conclude that $\Delta \le \Delta'$, and if such a coloring does not exist, the parties conclude that $\Delta > \Delta'$. The parties continue the binary search and output the coloring computed for the smallest value of $\Delta'$. This coloring is clearly a valid $(\Delta+1)$-coloring.

Each iteration of the binary search requires writing the conflict graph onto the blackboard, which requires $O(n \log^3 n)$ bits. Since the binary search requires $O(\log n)$ rounds, this protocol has communication cost $O(n \log^4 n)$ bits.

\section{Conclusion} \label{sec:conclusion}
In this paper, we gave a new lower bound technique that allowed us to prove $\Omega(n)$ lower bounds on the communication complexity of most graph problems in the $\bb(k)$ model, even if the edges are assigned uniformly at random to the parties. The strength of our approach is its wide applicability and the robustness under the random assignment of edges. We also showed that our lower bounds are tight, up to poly-logarithmic factors, for \textsf{Maximal Independent Set}, \textsf{Maximal Matching}, and $(\Delta+1)$-coloring.

We conclude with pointing out a connection between the two-party communication setting and the $\bb(k)$ model. It is not hard to see that a $\bb(k)$ model protocol $\mathcal{P}$ can be implemented in the two-party communication setting with the same (or less) communication cost. Hence, lower bounds in the two-party communication model translate to the $\bb(k)$ model. While some strong lower bounds in the two-party communication setting are known, such as the lower bound by Halld\'{o}rsson et al. \cite{hssw12} who showed that computing an $\alpha$-approximation to Maximum Independent Set requires communication cost $\tilde{\Omega}(n^2 / \alpha^2)$, these lower bounds do not hold under the random edge assignment assumption. This begs the following question: How can we prove strictly stronger lower bounds $\omega(n)$ for graph problems in the $\bb(k)$ model if edges are assigned uniformly at random to the parties?

\bibliographystyle{ACM-Reference-Format}
\bibliography{krz}

\end{document}